\documentclass[12pt]{article}

\usepackage{hyperref}
\usepackage{doi}
\usepackage[margin=0.8in]{geometry}
\usepackage{tikz-cd}
\usepackage[shortlabels]{enumitem}
\usepackage{xcolor}
\usepackage{bbold}

\usepackage[backend=biber, style=numeric, natbib=true, date=year, giveninits=true]{biblatex}
\AtEveryBibitem{\clearfield{issn}\clearfield{url}\clearfield{urlyear}\clearfield{urlmonth}\clearfield{urlday}\clearlist{language}}

\usepackage{amsmath, amssymb, amsthm, mathtools}
\usepackage{cmll}

\def\states{\mathfrak S}
\def\Ne{\mathcal N}
\def\Me{\mathcal M}
\def\Ee{\mathcal E}
\def\Fe{\mathcal F}
\def\Le{\mathcal L}
\def\Je{\mathcal J}
\def\Be{\mathcal B}
\def\Ae{\mathcal A}
\def\Se{\mathcal S}

\def\<{\langle}
\def\>{\rangle}
\def\Tr{\mathrm{Tr}}
\def\supp{\mathrm{supp}}

\newtheorem{lemma}{Lemma}
\newtheorem{theorem}{Theorem}
\newtheorem{coro}{Corollary}
\newtheorem{prop}{Proposition}

\newtheorem{introtheorem}{Theorem} 

\theoremstyle{definition}

\theoremstyle{remark}
\newtheorem{remark}{Remark}

\title{Sufficient positive maps between von Neumann algebras: R{\'e}nyi divergences}
\author{Anna Jen\v cov\'a$^1$, Lauritz van Luijk$^{2,3}$}

\date{\footnotesize{
$^1$Mathematical Institute of the Slovak Academy of Sciences, Bratislava, Slovakia\\
$^2$Perimeter Institute for Theoretical Physics, Waterloo, Ontario, Canada\\
$^3$Institute for Quantum Computing, Waterloo, Ontario, Canada\\
\today}}

\begin{document}

\maketitle

\begin{abstract}
We prove recovery theorems for $\alpha$-$z$ R\'enyi divergences under normal unital positive maps between von Neumann algebras. 
Previous results in this setting required 2-positivity, while recent finite-dimensional work showed that this assumption can be relaxed to mere positivity. Our proof uses sufficiency for JW*-subalgebras and $L^p$-spaces over them to characterize equality in the data processing inequality. 
As a further application of our methods, we solve a problem discussed by Haagerup and Størmer: Every conditional expectation of a von Neumann algebra onto a JW*-subalgebra factors through a conditional expectation onto the generated von Neumann subalgebra.
\end{abstract}

\tableofcontents

\null

\section{Introduction}

A core principle of quantum information theory and statistics is the data processing inequality, asserting that it is impossible for information-processing to increase how well a pair of states can be distinguished.
Distinguishability of states is quantified using divergences such as the quantum relative entropy or the various families of quantum R\'enyi divergences, each capturing different aspects of distinguishability.
In this paper we are concerned with quantum R\'enyi divergences.
More specifically, we consider the family of $\alpha$-$z$ quantum R\'enyi divergences
\cite{audenaert_-z-renyi_2015,hiai2024alphaz,kato2024onrenyi}, which contains
several established quantities, such as the sandwiched and  Petz quantum R\'enyi
divergences, see \cite{hiai2021quantum} for an overview of these special cases in the von Neumann algebraic setting.

It was famously shown by Petz \cite{petz_sufficiency_1988} that equality in the data processing inequality of the relative entropy holds if and only if the input states can be recovered from the output states.
Petz proved this for states on general von Neumann algebras with state transitions modeled by 2-positive maps.
Analogous recovery theorems have later been shown for various other families of
divergences \cite{jencova_reversibility_2012,jencova2017preservation,hiai2017different}.
In \cite{hiai2024alphaz}, it was shown to hold for the full family of $\alpha$-$z$ quantum R\'enyi divergences with state transitions modeled by $2$-positive maps between general von Neumann algebras. 

Until rather recently the data processing inequality of many important divergences was not known to hold for general positive maps.
This changed substantially in recent years
\cite{mullerhermes2017monotonicity,frenkel_integral_2023,hiai2024alphaz}.
For most divergences the data processing inequality is now known to hold for positive maps (see, e.g., \cite[Appendix E]{galke_sufficiency_2024} for a discussion), and there is no known example of a divergence that satisfies the data processing inequality for quantum channels (completely positive maps), but not for merely positive maps.
This poses the question whether the recovery theorems require 2-positivity, or whether they might hold under the minimal assumption of mere positivity?
Our work fits into a line of research on positive maps as generalized transformations of quantum statistical models \cite{alberti_problem_1980,mullerhermes2017monotonicity,galke_sufficiency_2024,frenkel_integral_2023,luijk2026sufficiency,luijk_quantum_2026}.

For states on finite-dimensional matrix algebras, recovery theorems for positive maps were finally shown for both the quantum relative entropy and the $\alpha$-$z$ quantum R\'enyi divergences in the paper \cite{luijk2026sufficiency}, which introduces Jordan algebraic techniques to handle positive maps in the same way that 2-positive maps are handled in earlier works. 
In this work, we focus on the case of R\'enyi divergences and prove recovery theorems for
positive maps between general von Neumann algebras by generalizing the argument in
\cite{luijk2026sufficiency}, which in turn rests on results developed in
\cite{jencova2018renyi,jencova2021renyi,hiai2024alphaz}. 

The $\alpha$-$z$ quantum R\'enyi divergences only define sensible divergences if the data processing inequality holds.
The region $R_{\text{DPI}}$ of pairs $(\alpha,z)$ for which this is the case was
characterized in \cite{zhang2020fromwyd}, and \cite{hiai2024alphaz} shows that the same
parameter region applies to the von Neumann algebraic setting, and the data processing inequality holds for general positive maps.
We denote by $R_{\text{rec}}\subset R_{\text{DPI}}$ the region for which the recovery
theorem holds for 2-positive maps, determined in \cite{hiai2024alphaz} (see Remark  \ref{rem:region} below, noting that this is indeed the full recovery region). We discuss this region explicitly in Section \ref{sec:recovery}.

To state our main result in a concise way, we assume faithfulness of all states involved
(this assumption will be dropped in the main part of the paper). 
In this case, the
recovery region only excludes the boundary points $\alpha=z=\frac12$ and $\alpha=z+1$ of $R_{\text{DPI}}$.

\begin{introtheorem}\label{introthm:main}
    Let $T:\Ne \to \Me$ be a normal unital positive map between von Neumann algebras $\Me$ and $\Ne$.
    Let $\rho, \omega$ be faithful normal states on $\Me$ such that $\rho\circ T, \omega\circ T$ are faithful states on $\Ne$.
    Let $(\alpha,z)\in R_{\text{rec}}$. 
    The following are equivalent:
    \begin{enumerate}
        \item Equality $D_{\alpha,z}(\rho\circ T\|\omega\circ T) = D_{\alpha,z}(\rho\|\omega)$,
        \item $T$ is sufficient for $\{\rho,\omega\}$: there exists a normal unital positive map $R:\Me\to \Ne$ such that $\rho\circ T\circ R=\rho$, $\omega\circ T\circ R = \omega$.
    \end{enumerate}
\end{introtheorem}

As in \cite{luijk2026sufficiency}, our proof strategy is to follow the arguments in
\cite{hiai2024alphaz} closely and to use Jordan algebraic techniques to take care of the
parts where \cite{hiai2024alphaz} uses the 2-positivity assumption.
More precisely, following \cite{luijk2026sufficiency}, we consider sufficiency of JW*-subalgebras $\Ne\subset \Me$ of von Neumann algebras, defined through the existence of a state-preserving normal unital positive map $T:\Me\to \Ne$.
A JW*-subalgebra is an ultraweakly closed, *-invariant, unital subspace of a von Neumann algebra that is closed under the Jordan product $(a,b)\mapsto \frac12(ab+ba)$.
For von Neumann subalgebras, this notion of sufficiency is shown to coincide with the traditional one \cite{luijk2026sufficiency,jencova2006sufficiencyfirst}, where $T$ is demanded to be 2-positive or, equivalently, completely positive.
As is the case for sufficiency of von Neumann algebras, there exists a minimal sufficient JW*-subalgebra $\Je_{\{\rho,\omega\}}$, which admits a state-preserving conditional expectation.
Here, a conditional expectation onto a JW*-subalgebra $\Je \subset \Me$ means a normal unital positive projection onto $\Je$, satisfying a Jordan version of the usual bimodule property.
The known fact that a conditional expectation is completely positive if and only if its range is a von Neumann subalgebra implies that a von Neumann subalgebra is sufficient in the sense of positive maps if and only if it is sufficient in the sense of completely positive maps.
Moreover, it follows that the minimal sufficient von Neumann subalgebra $\Ae_{\{\rho,\omega\}}$ is the von Neumann algebra generated by the minimal sufficient JW*-subalgebra $\Je_{\{\rho,\omega\}}$, i.e.,
\begin{equation}
    \Ae_{\{\rho,\omega\}} = \Je_{\{\rho,\omega\}}''.
\end{equation}
These results are not restricted to the binary case, instead they apply to general faithful statistical experiments $\Ee$ (faithfulness means that the family of states is faithful, not that each state $\omega\in\Ee$ is faithful).

A further ingredient in our proofs is a suitable notion of $L^p$-spaces over JW*-subalgebras $\Je\subset \Me$. 
In contrast to previous works
\cite{ayupov1984integration,arhancet_spectral_2026,arhancet_completely_2026,arhancet_positive_2023},
we do not work with abstract $L^p$-spaces associated with $\Je$. 
Instead, we use a conditional expectation onto $\Je$ to realize the $L^p$-spaces as subspaces of the Haagerup $L^p$-spaces of $\Me$.
We establish the basic properties needed for our arguments, including their behavior under multiplication and conditional expectations. 
As in \cite{luijk2026sufficiency}, these results allow us to adapt $L^p$-space techniques
previously used in the 2-positive setting \cite{hiai2024alphaz,jencova2018renyi,jencova2021renyi} to merely positive maps.

Next, we explain how the theory of sufficient subalgebras gives a complete and elegant answer to a question raised and partially answered by Haagerup and Størmer in \cite{stormer_positive_1963}.
The question is: 
Does every faithful conditional expectation $E: \Me \to \Je$ onto a JW*-subalgebra $\Je$ of a von Neumann algebra $\Me$ factor through a conditional expectation onto the von Neumann subalgebra $\Je''$ generated by $\Je$?
Using the classification of JW*-algebras and their representations, they were able to settle this affirmatively in some special cases.
If one considers the family $\Ee$ of $E$-invariant states, it is easy to see that the minimal sufficient JW*-subalgebra is $\Je_\Ee=\Je$ and, therefore, the minimal sufficient von Neumann subalgebra is $\Je''$.
Thus, there exists a conditional expectation $F$ of $\Me$ onto $\Je''$ preserving every $E$-invariant state.
Hence, $E|_{\Je''} \circ F$ is a conditional expectation onto $\Je$.
Since there can be at most one conditional expectation preserving a fixed faithful normal state, this conditional expectation equals $E$.
Therefore, we have:

\begin{introtheorem}\label{introthm:CE}
    Let $\Me$ be a $\sigma$-finite von Neumann algebra.
    Let $E: \Me \to \Je$ be a faithful conditional expectation onto a JW*-subalgebra $\Je \subset \Me$.
    Then $E$ factorizes through a faithful conditional expectation $F: \Me \to \Je''$ onto the generated von Neumann algebra, i.e.,
    \begin{equation}
        E = E|_{\Je''} \circ F.
    \end{equation}
\end{introtheorem}

This result, and the ensuing general structure of faithful conditional expectations on
$\sigma$-finite von Neumann algebras, interesting in its own right, provides one of the
crucial tools for this work.

\paragraph{Acknowledgements.}
We thank Henrik Wilming and Jonas Bettermann for many helpful discussions. 
This work was supported by the Swedish Research Council under grant no.\ 2021-06594 while the authors were at Institut Mittag-Leffler in Djursholm, Sweden, during the 2026 program on Operator Algebras and Quantum Information.
AJ was supported by the grant VEGA 2/0128/24. 
Research at Perimeter Institute and the University of Waterloo is supported in part by the
Government of Canada through the Department of Innovation, Science and Economic
Development and by the Province of Ontario through the Ministry of Colleges and
Universities.

\section{Definitions}

For a von Neumann algebra $\Me$, we will denote  by $\Me_*$ the predual of $\Me$  and
by $\Me_*^+$ its positive part. For $0<p\le \infty$, we denote by $L^p(\Me)$ the Haagerup
$L^p$-space over $\Me$ \cite{haagerup1979lpspaces}. 
We refer to \cite{terp1981lpspaces}
and \cite[Sec.~9]{hiai2021lectures} for the necessary definitions and properties. The cone
of positive elements will be denoted by $L^p(\Me)^+$ and the real subspace of self adjoint
elements by $L_h^p(\Me)$.

For any normal functional $\varphi\in \Me_*$, we will denote by $h_\varphi$ the operator in
$L^1(\Me)$ corresponding to $\varphi$ under the linear and order isomorphism $\Me_*\simeq
L^1(\Me)$. The trace functional in $L^1(\Me)$ will be denoted by $\Tr$, that is, $\Tr
h_\varphi=\varphi(1)$, $\varphi\in \Me_*$.

\subsection{Normal unital positive maps and their Petz duals}

We consider normal unital positive (NUP) maps $T:\Ne\to \Me$ between von Neumann algebras $\Me$ and $\Ne$.
Under the identification of the preduals $\Me_*$, $\Ne_*$ with the Haagerup $L^1$-spaces $L^1(\Me)$, $L^1(\Ne)$, the preadjoint of an NUP map is a positive trace-preserving map $T_*:L^1(\Me)\to L^1(\Ne)$.
The map $T$ is  completely positive if and only if the predual $T_*$ is.
We abbreviate completely positive normal unital positive maps as NUCP maps.



Let $T:\Ne\to \Me$ be an NUP map and let $\omega$ be a faithful normal state. 
Let $\omega_0:=\omega\circ T$ and assume that $\omega_0$ is faithful as well. 
If $\omega_0$ is not faithful, we can restrict our attention to the corner of the support projection to obtain a faithful state. 
Let $p>0$, then the map
\begin{equation}\label{eq:Tp}
h_{\omega_0}^{\frac{1}{2p}}ah_{\omega_0}^{\frac{1}{2p}}\mapsto
h_\omega^{\frac1{2p}}T(a)h_\omega^{\frac1{2p}},\qquad a\in \Ne
\end{equation}
extends to a linear map $T_p: L^p(\Ne)\to L^p(\Me)$. By the identification
$L^\infty(\Me)=\Me$, we may put $T_\infty:=T$. 
\begin{prop}\label{prop:contraction}\cite[Lemma 3.1]{hiai2024alphaz} For $p\ge 1$, $T_p$ is a contraction.

\end{prop}

By the results in \cite{jencova2018renyi}, there exists an NUP map $T_\omega: \Me\to \Ne$ such that
\[
T_*(h^{\frac12}_\omega x h^{\frac12}_\omega)= h_{\omega_0}^{\frac12}T_\omega(x)
h_{\omega_0}^{\frac12}=(T_\omega)_1(h^{\frac12}_\omega x h^{\frac12}_\omega).
\]
This is precisely the Petz dual of $T$, or the Petz recovery map \cite{petz_dual_1984}. This map satisfies
$\omega_0 \circ T_\omega=\omega$, and $(T_\omega)_{\omega_0}=T$. Moreover, $T_\omega$ is
$n$-positive if and only if $T$ is. 
The above equality shows that $T_*=(T_\infty)_*=(T_\omega)_1$. For $1\le p<\infty$ and $\frac1p+\frac1q=1$,  we have 
\begin{equation}\label{eq:tpq}
T_p^*=(T_\omega)_q.
\end{equation}
In particular, for $p=2$, $T_2$ defines a contraction 
$L^2(\Ne) \to L^2(\Me)$ whose adjoint is $T_2^*=(T_\omega)_2$.

\subsection{JW*-subalgebras and conditional expectations}\label{sec:prelim-CEs}

A JW*-subalgebra of a von Neumann algebra $\Me$ is a unital ultraweakly closed *-invariant subspace $\Je\subset \Me$ that is closed under the Jordan product
\begin{equation}
    a\circ b= \frac12(ab+ba).
\end{equation}
(In fact, it is enough to ask that $\Je$ be closed under taking squares, i.e., $a\in \Je\implies a^2\in \Je$ \cite{hanche-olsen_jordan_1984}.)
A JW*-subalgebra is automatically closed under the Jordan triple product
\begin{equation}\label{eq:triple-prod}
    \{abc\} = a \circ (b\circ c) + (a\circ b)\circ c - b \circ(a\circ c) = \frac12(abc+cba), \qquad a,b,c \in \Me.
\end{equation}
Note that $\{ab1\} = a\circ b$ and $\{aba\} = aba$, $a,b\in \Me$.
If $\Je$ is closed under higher-order symmetrized products, it is called a reversible JW*-subalgebra.
We refer to \cite{hanche-olsen_jordan_1984} for details on JW*-algebras.

Let $\Je\subset \Me$ be a JW*-subalgebra of a von Neumann algebra, and let $\Ae=\Je''$ be the von Neumann algebra it generates.
Then, every central projection in $\Je$ is central in $\Ae$. Thus, every direct sum decomposition of $\Je$ induces a direct sum decomposition of $\Ae$. 
By \cite{haagerup1995positive}, one can decompose $\Je$ as a direct sum of the following types:
\begin{enumerate}[(i)]
    \item\label{it:case-i} $\Je$ is a von Neumann subalgebra. In this case $\Ae =\Je$.
    \item\label{it:case-ii} $\Je$ is reversible. This case decomposes as a direct sum of two subcases:
    \begin{enumerate}[({ii.}a)]
        \item\label{it:case-iia} $Z(\Ae) \ne Z(\Je)$. In this case, there exists a von Neumann algebra $\Be$ and an isomorphism $\Ae \simeq \Be \oplus \Be^{op}$ such that 
        \begin{equation}
            \Je \simeq \{ a\oplus a^{op} \ : a\in \Be\}.
        \end{equation}
        \item\label{it:case-iib} 
        $Z(\Je)=Z(\Ae)$. In this case, $\Je \subset \Ae$ is the set of fixed points of an involutive *-antiautomorphism $\alpha$ on $\Ae$.
    \end{enumerate}
    \item\label{it:case-iii} $\Je$ is not reversible. 
\end{enumerate}
This decomposition will be very useful for us, as it will allow us to reduce statements about JW*-algebras to statements about von Neumann algebras.

We define a (normal) conditional expectation $E$ on a von Neumann algebra $\Me$ as a normal unital positive map $E:\Me \to \Me$, which is idempotent, i.e., $E^2=E$, and satisfies the following Jordan bimodule property
\begin{equation}\label{eq:J-bimodule}
    E(\{abc\}) = \{aE(b)c\}, \qquad a,c \in \Je,\ b \in \Me,
\end{equation}
where $\Je = E(\Me)$ denotes the range of $E$.
It immediately follows that $\Je$ is a JW*-subalgebra: 
if $a,b \in \Je$, then $a\circ b = \{aE(1)b\} = E(\{a1b\}) = E(a\circ b)$.
We will say that $E$ is a conditional expectation onto $\Je$.

It can be shown that a conditional expectation $E$ is completely positive if and only if its range $\Je$ is a von Neumann subalgebra if and only if the full bimodule property $E(abc)=aE(b)c$, $a,c\in \Je$, $b\in \Me$, holds \cite{tomiyama_projection_1957,choi_schwarz_1974}.

For faithful maps, the Jordan bimodule property already follows from idempotence:

\begin{lemma}
    Let $E:\Me\to \Me$ be a faithful normal unital positive map.
    Then, $E$ is a conditional expectation if and only if it is idempotent: $E^2=E$.
\end{lemma}

\begin{proof}
    \cite[Prop.~3.21]{arhancet_positive_2023} shows that $\Je = E(\Me)$ is a JW*-subalgebra and that $E(a\circ b)=a\circ E(b)$ for all $a\in\Je$, $b\in\Me$.
    This and \eqref{eq:triple-prod} yield the claim: 
    For $a,c\in\Je$ and $b\in\Me$,
    \begin{align*}
        E(\{abc\}) 
        &= E(a\circ(b\circ c))- E(b\circ (a\circ c)) + E(c\circ (b\circ a))\\
        &=a \circ (E(b)\circ c) - E(b)\circ(a\circ c) + c \circ (a\circ E(b)) 
        = \{aE(b)c\}.
    \end{align*}
    
\end{proof}

Let $\omega$ be a normal state on $\Me$ and let $E$ be an $\omega$-preserving  conditional
expectation, that is $\omega\circ E=\omega$. It is easily seen that $E$ is faithful if  $\omega$ is.
In that case, $E$ is the unique $\omega$-preserving conditional expectation onto its
range,  \cite[Lemma 1]{luczak2021some}.

%
%
%

We will need the following lemma about the $L^2$-realization of state-preserving conditional expectations.

\begin{lemma}\label{lemma:idempotent} Let $E:\Me\to \Me$ be a conditional expectation.
Let $\omega$ be a normal faithful state such that $\omega\circ E=\omega$. 
We then have $E_\omega=E$, and $E_2$ is an
 orthogonal projection on $L^2(\Me)$. 

\end{lemma}

\begin{proof} Since $E$ is idempotent and $\omega\circ E=\omega$, it is clear that $E_2$ is
idempotent as well. Since $E_2$ is also a contraction on the Hilbert space $L^2(\Me)$ by Proposition
\ref{prop:contraction}, it follows that $E_2$ must be an orthogonal projection. Hence
$E_2=E_2^*=(E_\omega)_2$, so that 
\[
h_{\omega}^{\frac14}E(x)h_\omega^{\frac14}=h_{\omega}^{\frac14}E_\omega(x)h_\omega^{\frac14}
\]
by definition of $T_2$. 
Since $\omega$ is faithful, this implies that $E_\omega=E$.

\end{proof}

\section{Statistical experiments and sufficient maps}

A statistical experiment is a pair $\Ee=(\Me,\Se)$, where  $\Me$ is a von Neumann algebra and   $\Se:=\{\rho_\theta\}_{\theta\in
\Theta}$ is a set of normal states on  $\Me$. 
We will say that the experiment is faithful if, for $0\le a \in \Me$, $\varphi(a)=0$ for all $\varphi\in \Se$ implies $a=0$.
The support projection of the experiment $\Ee$ is defined as
\[
\supp(\Se) := \vee_{\omega\in\Se}\,\supp(\omega).
\] 
Faithfulness is equivalent to having $\supp(\Se)=1$. Following the terminology of
\cite{jencova2006sufficiencyfirst}, we say that a normal state $\omega$ dominates $\Se$ if, for $0\le a\in \Me$, $\omega(a)=0$ implies $\varphi(a)=0$ for all $\varphi\in \Se$, or, equivalently, if $\supp(\Se)\le \supp(\omega)$.
In this case, the dominating state may be chosen from the closed convex hull of $\Se$.
For our purposes, we may always make this choice. Moreover,
we will always assume in
the sequel that $\Me$ is $\sigma$-finite, in which case a dominating state always exists.

We will say that an NUP map $T:\Ne\to \Me$ is sufficient with respect to the experiment
$(\Me,\Se)$ if there exists an NUP map $S:\Me\to\Ne$ such that 
\[
\rho\circ T\circ S=\rho,\qquad \forall \rho\in \Se.
\]
The NUP map $S$ is called a recovery map for $T$. Moreover, if the recovery map can be chosen to be completely positive, i.e., an NUCP map, we will say that $T$ is cp-sufficient with respect to $(\Me,\Se)$.

If $T:\Ne\to \Me$ is an NUP map and $\omega$ is a state dominating $\Se$, then it is easy
to see that $\omega\circ T$ dominates
$\{\rho\circ T\colon \rho\in \Se\}$. By standard arguments (see e.g.
\cite{luijk2026sufficiency}), when considering (cp-)sufficiency of $T$ with respect to
$\Ee$, we may assume without loss of generality that $\Se$ is dominated by a faithful
state  $\omega$ such that  $\omega\circ T$ is also faithful. 
Further, since $\omega$
can be chosen from the closed convex hull of $\Se$, we may assume that $\omega$ is itself
contained in $\Se$.

If $\Me_0\subseteq \Me$ is a JW*-subalgebra, we will say that $\Me_0$ is (cp-)sufficient with
respect to $(\Me,\Se)$ if there is an NU(C)P map $S:\Me\to \Me$ with range in $\Me_0$
such that 
\[
\rho\circ S=\rho,\qquad \rho\in \Se.
\]
Note that if $\Me_0$ is in fact a von Neumann subalgebra, then it is (cp-)sufficient if and only if 
the inclusion map $\Me_0\hookrightarrow \Me$ is (cp-)sufficient with respect to
$(\Me,\Se)$. 
Also note that the inclusion map $\Me_0\hookrightarrow \Me$ is completely positive if and only if $\Me_0$ is a von Neumann subalgebra.

\subsection{Minimal sufficient subalgebras and universal recovery}

Let $\Ee=(\Me,\Se)$ be a statistical experiment, with a faithful normal state $\omega\in
\Se$. 

We say that a JW*-subalgebra  $\Me_0\subseteq \Me$ is minimal sufficient
with respect to $\Ee$ 
if it is sufficient and contained in any sufficient JW*-subalgebra. Note that a
minimal sufficient JW*-subalgebra is necessarily unique, if it exists.
In this section, we will prove existence of minimal sufficient JW*-subalgebras. We will also prove that the Petz recovery map is a universal recovery map. 
Both these statements follow from an argument using ergodic theory, which we adapt from
the works by \L uczak \cite{luczak2014quantum,luczak2021some}.

Let $\Le$ denote the set of NUP maps $R:\Me\to \Me$, such that $\rho\circ R=\rho$ for all
$\rho\in \Se$. It is easy to see that $\Le$ is a semigroup, moreover, it is convex and
closed in the point ultraweak topology. By the mean ergodic theorem \cite{thomsen1985invariant}, it follows that
there is an idempotent map $E\in \Le$ such that
\[
R\circ E=E\circ R=E,\qquad \forall R\in \Le.
\]
Moreover, the range $\Je_\Ee$  of $E$ is the set of common fixed points of all the maps in $\Le$. 

\begin{prop} \label{prop:min_suf} 
    $\Je_\Ee$ is the  minimal sufficient JW*-subalgebra with respect to
$\Ee$.
\end{prop}

\begin{proof} Note that $E\in \Le$, so that it is an idempotent positive unital map on
$\Me$ 
preserving the faithful state $\omega$. It follows that $E$ is a faithful normal conditional expectation
and its range $\Je_\Ee$ is a JW*-subalgebra in $\Me$. Since also $\rho\circ
E=\rho$ for all $\rho\in \Se$, $\Je_\Ee$ is sufficient. If $\Je$ is any  sufficient
JW*-subalgebra with respect to $\Ee$, then there is some NUP map $R:\Me\to\Me$ with range in
$\Je$ such that $\rho\circ R=\rho$ for all $\rho\in \Se$, in other words, $R\in \Le$. It
follows that $E\circ R=R\circ E=E$, so that we have 
\[
a=E(a)=R\circ E(a),\qquad \forall a\in \Je_\Ee.
\]
This shows that $\Je_\Ee\subseteq \mathrm{Range}(R)\subseteq \Je$, so that $\Je_\Ee$ is minimal sufficient.

\end{proof}

We now show that the Petz recovery map is universal.

\begin{prop}\label{prop:universal} 
An NUP map  $T:\Ne\to \Me$ is  sufficient with respect
to $\Ee$ if and only if
\[
\rho\circ T\circ T_\omega=\rho,\qquad \forall \rho\in \Se.
\]

\end{prop}

\begin{proof} Assume that $T$ is sufficient and let $S:\Me\to \Ne$ be a recovery map, then
clearly $T\circ S\in \Le$, and therefore
\[
T\circ S\circ E=E\circ T\circ S=E.
\]
Using the representation as contractions on $L^2(\Me)$ with the faithful state $\omega$, we obtain
\[
S_2^*T_2^*E_2=E_2,
\]
since $E_2$ is an orthogonal projection on $L^2(\Me)$, by Lemma \ref{lemma:idempotent}. 
Since both $T_2$ and $S_2$ are contractions, we have 
\[
\|\xi\|_2=\|S_2^*T_2^*\xi\|_2\le \|T_2^*\xi\|_2\le \|\xi\|_2,
\]
and hence  $\|T_2^*\xi\|_2=\|\xi\|_2$, for any $\xi$ in the range of $E_2$. This is
equivalent to $T_2T_2^*E_2=E_2$. 
Since $E=E_2^*$ and $T_2T_2^*$ are self-adjoint, we also have $E_2T_2T_2^*=E_2$, which shows $E\circ T\circ T_\omega=E$. 
But then for
any $\rho\in \Se$, we get
\[
\rho\circ T\circ T_\omega=\rho\circ E\circ T\circ T_\omega=\rho\circ E=\rho.
\]
The converse statement is clear.

\end{proof}

\begin{coro}\label{coro:cpsuf} Let $T$ be an NUCP map. Then $T$ is sufficient with
respect to $\Ee$ if and only if it is cp-sufficient. In particular, a sufficient von Neumann
subalgebra in $\Me$ is cp-sufficient. 

\end{coro}

\begin{proof} 
Assume that $T$ is sufficient with respect to $\Ee$. 
By Proposition  \ref{prop:universal}, $T_\omega$ is a recovery map for $T$. 
Since $T_\omega$ is completely positive if $T$ is, the claim follows. 

\end{proof}

\begin{coro}\label{coro:suff_ext} Let $T:\Ne\to \Me$  be  sufficient with respect to
$\Ee$, and let $E$ be the conditional expectation onto $\Je_\Ee$. Then $T$  is sufficient for $\Ee'=(\Me,\Se')$, where $\Se'$ is any set of
normal states that are invariant with respect to $E$.
\end{coro}

\begin{proof} Assume that $T$ is sufficient with respect to $\Ee$ and let $S$ be any NUP recovery map. Then 
\[
\rho\circ T\circ S=\rho\circ E\circ T\circ S=\rho\circ E=\rho,\qquad \forall \rho\in \Se',
\]
hence $T$ is sufficient with respect to $\Ee'$.

\end{proof}

\subsection{Minimal sufficient von Neumann subalgebras}

Let $\Ae\subseteq \Me$ be a von Neumann subalgebra. Similarly as before, we say that
$\Ae$ is a minimal sufficient von Neumann subalgebra if it is sufficient and is
contained in any other sufficient von Neumann subalgebra. By Corollary
\ref{coro:cpsuf}, a minimal sufficient von Neumann subalgebra is in fact minimal cp-sufficient and
vice versa, so that these two possible notions coincide.
Note that a minimal sufficient von Neumann subalgebra is necessarily unique, if it exists.
The following statement is an immediate consequence of \cite[Thm.~1]{luczak2014quantum}.

\begin{theorem}\label{thm:minsuf} 
Let $\Ae_\Ee$ denote  the von Neumann subalgebra generated by
$\Je_\Ee$. Then $\Ae_\Ee$ is the minimal sufficient von Neumann subalgebra with respect to $\Ee$.
The subalgebra $\Ae_\Ee$ is the range of a faithful normal completely positive conditional expectation
$F:\Me\to \Me$ and we have the unique factorization 
\[
E=\tilde E\circ F,
\]
where $\tilde E$ is a conditional expectation of $\Ae_\Ee$ onto $\Je_\Ee$. 
\end{theorem}

\begin{proof} It follows from the proof of \cite[Thm.~1]{luczak2014quantum} that $\Ae_\Ee$
is the minimal sufficient von Neumann subalgebra with respect to $\Ee$. By Corollary
\ref{coro:cpsuf}, it is in fact minimal cp-sufficient.
There is  an alternative construction of the minimal cp-sufficient subalgebra,
analogous to the proof of Proposition \ref{prop:min_suf}.  
We again use the mean ergodic theorem, this time with the subsemigroup  $\Le_{cp}$ consisting of all NUCP maps in $\Le$.  
The minimal idempotent in this case will be
denoted by $F$. Note that $F$ is a completely positive conditional expectation, and therefore its range is
a von Neumann subalgebra.  It can be  proved exactly the same  way as Proposition \ref{prop:min_suf} that
the range of $F$ is minimal cp-sufficient, so that we must have $F(\Me)=\Ae_\Ee$ by
uniqueness.

Since  $F\in \Le_{cp}\subseteq \Le$, we have $E\circ F=F\circ E=E$.  Put $\tilde
E=E|_{\Ae_{\Ee}}$, then $\tilde E$ is a faithful normal conditional expectation on $\Ae_\Ee$ with
range $\Je_\Ee$ and we get
\[
E=E\circ F=\tilde E\circ F.
\]
Uniqueness is clear, since $F$ must preserve the faithful state $\omega$ and similarly $\tilde E$ preserves $\omega|_{\Ae_\Ee}$.

\end{proof}

\section{Conditional expectations onto JW*-algebras}

Let  $E:\Me\to \Me$ be a faithful conditional expectation with range $\Je$. 
Haagerup and  Størmer \cite{haagerup1995positive} asked the question whether $E$
factorizes through the generated von Neumann subalgebra $\Ae=\Je''$ as in the following diagram:
\begin{equation}
    \begin{tikzcd}
    \Me \arrow{rr}{E} \arrow{dr}{F} & & \Je \\
    &\Ae \arrow{ur}{\tilde E}
\end{tikzcd}
\end{equation}
Here $F$ and $\tilde E =E|_\Ae$ are conditional expectations. This question was solved in
\cite{haagerup1995positive} in some special cases. We next observe that if $\Me$ is
$\sigma$-finite,  existence of the factorization is a rather straightforward consequence
of Theorem \ref{thm:minsuf}.

\begin{coro}\label{coro:positive_projections} Let $\Me$ be a $\sigma$-finite von Neumann algebra and let
$E:\Me\to \Me$ be a faithful normal conditional expectation onto a JW*-subalgebra $\Je\subset \Me$.
Let $\Ae$ be the von Neumann subalgebra generated by $\Je$. Then there
exists a unique faithful normal conditional expectation $F$ onto  $\Ae$ 
and a unique faithful normal conditional expectation $\tilde E =E |_{\Ae}$ on $\Ae$ with range
$\Je$, such that $E=\tilde E\circ F$.
\end{coro}

\begin{proof} Let us consider the experiment $\Ee_E=(\Me,\Se_E:=\{\rho\circ E: \rho\in
\states(\Me)\})$, where $\states(\Me)$ denotes the set of normal states of $\Me$.  Since $\Me$ is $\sigma$-finite and $E$ is faithful, there is some
faithful normal state $\omega\in \Se_E$. It is now easy to see that $\Je=\Je_{\Ee_E}$.
Indeed, let $E'$ be the minimal idempotent for the experiment $\Ee_E$, then since $E$
preserves all the states in $\Se_E$, we have $E'\circ E=E\circ E'=E'$. On the other hand,
this implies that 
 \[
\rho\circ E=\rho\circ E\circ E'=\rho\circ E',\qquad \rho\in \states(\Me),
 \]
which shows that $E=E'$, and consequently $\Je=\Je_{\Ee_E}$, $\Ae=\Ae_{\Ee_E}$. 
The rest follows by Theorem \ref{thm:minsuf}.

\end{proof}

\subsection{The structure of conditional expectations} \label{sec:structure}

We will now discuss the general form of conditional expectations $\tilde E:\Ae \to \Je$ of a von Neumann algebra $\Ae$ onto a generating JW*-subalgebra $\Je\subset \Ae$.
As discussed in Sect.~\ref{sec:prelim-CEs}, $\Je \subset \Ae$ can be decomposed as a direct sum of inclusions of the following types: \ref{it:case-i} $\Je$ is itself a von Neumann algebra, \ref{it:case-iia} $\Je$ is a reversible subalgebra with $Z(\Je)\ne Z(\Ae)$, \ref{it:case-iib} $\Je$ is a reversible subalgebra with $Z(\Je)=Z(\Ae)$, and \ref{it:case-iii} $\Je$ is not reversible.
The Jordan bimodule property (see \eqref{eq:J-bimodule}) implies that a conditional expectation $\tilde E:\Ae \to \Je$ is compatible with the direct sum decomposition, i.e., it decomposes as a direct sum of conditional expectations in each of the four cases.
Therefore, we only have to understand what conditional expectations look like in each of these cases.

\textbf{Case \ref{it:case-i}.}~
Since $\Je$ generates $\Ae$, we have $\Je=\Ae$, so the only conditional expectation $E:\Ae\to\Je$ is the identity.

\textbf{Case \ref{it:case-iia}.}~
We may suppress the isomorphism and assume $\Ae = \Be\oplus \Be^{op}$ and $\Je = \{a\oplus a^{op} : a\in \Be\}$ for a von Neumann algebra $\Be$.
Suppose $\tilde E:\Ae\to \Je$ is a conditional expectation.
Let $p=1\oplus 0\in \Ae$ and let $s$ be any symmetry in $\Je$. 
Then by the Jordan bimodule property of
$\tilde E$, we have
\[
s\tilde E(p)s=\tilde E(sps)=\tilde E(p),
\]
since $p$ is central in $\Ae$. It follows that $\tilde E(p)\in Z(\Je)$, so that there is
some positive element $\lambda\le 1$ in the center of $\Be$ such that $\tilde
E(p)=\lambda\oplus\lambda$, $\tilde E(0\oplus 1)=\tilde E(1-p)=(1-\lambda)\oplus
(1-\lambda)$. 
{Here, we identify $\lambda$ with $\lambda^{op}$ since $Z(\Be^{op})\cong Z(\Be)^{op} = Z(\Be)$.}
Further, for any $a\in \Be$, we have
\[
\tilde E(a\oplus 0)=\tilde E(p\circ (a\oplus a^{op}))=\tilde E(p)\circ (a\oplus a^{op})=\lambda
a\oplus \lambda a^{op},
\]
and similarly for $\tilde E(0\oplus b^{op})$, so that
\begin{equation}\label{eq:CEs-iia}
    \tilde E(a\oplus b^{op})=(\lambda a+(1-\lambda)b)\oplus (\lambda a^{op}+(1-\lambda)b^{op}).
\end{equation}
In turn, \eqref{eq:CEs-iia} defines a conditional expectation for every element $0\le \lambda\le 1$ in $Z(\Be)$.
Hence, the possible conditional expectations are given by \eqref{eq:CEs-iia}, and are parametrized by the set of effects in $Z(\Be)$. 
The conditional expectation $\tilde E$ given by \eqref{eq:CEs-iia} is faithful precisely when $\supp(\lambda) = \supp(1-\lambda)=1$.

\textbf{Case \ref{it:case-iib}.}~
In this case, there is an involutive *-antiisomorphism $\alpha$ on $\Ae$ such that $\Je$ is its fixed-point set.
Clearly a faithful conditional expectation 
$\tilde E:\Ae \to \Je$ is given by
\[
\tilde E=\frac12(id+ \alpha).
\]
According to \cite[Theorem 3.2]{haagerup1995positive}, this is the unique conditional expectation of $\Ae$ onto $\Je$.

\textbf{Case \ref{it:case-iii}.}~ 
If $\Je$ is not reversible, it is of type $I_2$. It is shown in
\cite[Thm.~2.1]{haagerup1995positive} that in this case a conditional expectation of $\Ae$ onto $\Je$ exists if and only if $\Ae$ is a finite von Neumann algebra.
If $\Ae$ is finite, every normal faithful tracial state $\tau$ on $\Ae$ defines a
conditional expectation onto $\Je$ via
\begin{equation}\label{eq:tracial}
    \tau(x\circ a)= \tau(\tilde E(x)\circ a),
    \qquad x\in \Ae,\  a\in \Je.
\end{equation}
It is easy to see that this is the unique $\tau$-preserving conditional expectation.
In fact, as shown in \cite{haagerup1995positive}, every faithful conditional expectation $\tilde E$ of $\Ae$ onto $\Je$ is of this form for some normal faithful tracial state $\tau$.


\subsection{$L^p$-spaces over JW*-subalgebras}

In this section, we will consider $L^p$-spaces over a  JW*-algebra $\Je$,  under the assumptions that
we keep throughout this paper. 
Specifically, we assume that $\Je$  is the range of a conditional
expectation $E$ on a $\sigma$-finite von Neumann algebra $\Me$,  with a faithful normal invariant state $\omega$. For any $p>0$,
we define the space $L^p(\Je,\omega)$ as the closure of the subspace
\begin{equation}
    h_\omega^{\frac{1}{2p}} \Je h_\omega^{\frac{1}{2p}} \subset L^p(\Me)
\end{equation}
in the corresponding (quasi) norm. For $p=\infty$, we put $L^p(\Je,\omega)=\Je$. 
Let us also denote by $L^p_h(\Je,\omega)$ the set of self-adjoint elements in
$L^p(\Je,\omega)$. We then have the following properties.

\begin{lemma}\label{lemma:lpJ}
\begin{enumerate}[(a)]
\item\label{it:lpJ1} 
$L^1(\Je,\omega)=\{h_\varphi\in L^1(\Me)\colon \varphi\circ E=\varphi\}$.
\item\label{it:lpJ2} For any $p>0$,
\[
L^p_h(\Je,\omega)=\{h=h^*\in L^p(\Me)\colon \text{ with polar decomposition } s|h|,\ s\in
\Je,\ |h|^p\in L^1(\Je,\omega)\}.
\]
\item\label{it:lpJ3} If $\frac1s = \frac1p+\frac1q+\frac1r$ with $p,q,r>0$, then 
\[
\{hkl\}\in L^s(\Je,\omega),\qquad \forall\ h\in L^p(\Je,\omega),\ k\in L^q(\Je,\omega),\ l
\in L^r(\Je,\omega).
\]
\item\label{it:lpJ4} For $p,q>0$, $k\in L^p(\Me)$ and $h=h_\mu^{\frac1{2q}}$ for a faithful
state $\mu$ invariant under $E$, we have $hkh\in
L^r(\Je,\omega)$ if and only if $k\in L^p(\Je,\omega)$, with $\frac1p+\frac1q=\frac1r$.
\item \label{it:lpJ5} For $p\ge 1$, $L^p(\Je,\omega)$ is the range of the contractive projection $E_p$ on 
$L^p(\Me)$, defined by \eqref{eq:Tp}. With
$\frac1s = \frac1p+\frac1q+\frac1r$,  $p,q,r,s \in [1,\infty]$, we have the following
extension of the Jordan bimodule
property 
\begin{equation}\label{eq:p_bimodule}
    E_s(\{hkl\}) = \{h E_q(k) l\}, \qquad h\in L^p(\Je,\omega),\ k\in L^q(\Me),\ l \in L^r(\Je,\omega).
\end{equation}

\end{enumerate}

\end{lemma}

Notice that \ref{it:lpJ1} and \ref{it:lpJ2} in the above lemma imply that
$L^p(\Je,\omega)$ only depends on the conditional expectation $E$, that is,
$L^p(\Je,\omega)=L^p(\Je,\mu)$ for any faithful normal state $\mu$ invariant under $E$.

\begin{proof} Notice first that for $p\ge 1$, the map $E_p$ is a projection on $L^p(\Me)$,
with range equal to $L^p(\Je,\omega)$.
Indeed,  $h_\omega^{\frac{1}{2p}} \Je h_\omega^{\frac{1}{2p}}$ is
clearly included in the range of $E_p$, so that $L^p(
\Je,\omega)\subseteq E_p(L^p(\Me))$. For the converse, let $x\in L^p(\Me)$, then there is a
sequence $(y_n)_n\subset \Me$ with $x=\lim_n h_\omega^{\frac1{2p}}y_nh_{\omega}^{\frac1{2p}}$, so that, by \eqref{eq:Tp}, 
$E_p(x)=\lim h_{\omega}^{\frac1{2p}} E(y_n) h_{\omega}^{\frac1{2p}} \in L^p(\Je,\omega)$. 
In particular, since $E_1=E_*$, this implies the statement \ref{it:lpJ1}. 

The strategy for the rest of the proof is as follows. 
We first observe that if $\Je$ is a von Neumann subalgebra (that is, $E$ is completely
positive), then we may identify $L^p(\Je,\omega)\simeq L^p(\Je)$, and all the statements hold in this case.
Thus, using the decomposition $E=\tilde E\circ F$ in Corollary \ref{coro:positive_projections}, we may reduce to the case when $\Me=\Je''$.
We then consider the decomposition $\Je=\Je_{\text{(i)}}\oplus\Je_{\text{(ii.a)}}\oplus \Je_{\text{(ii.b)}} \oplus \Je_{\text{(iii)}}$ as in Section \ref{sec:prelim-CEs}.
Since the central projections of $\Je$ are also central in $\Me$, we have a corresponding direct sum decomposition of $\Me$, and hence of $L^p(\Me)$, with $E_p$ acting separately on each of the components. 
It is therefore enough to consider the four cases \ref{it:case-i}, \ref{it:case-iia}, \ref{it:case-iib}, and \ref{it:case-iii} separately.
In case \ref{it:case-i}, $\Je$ is a von Neumann subalgebra, so we will have already solved this case.
For the remaining cases, we use the structure of conditional expectations determined in  Section \ref{sec:structure}. 

So assume that $\Je$ is a von Neumann subalgebra in $\Me$. Since $\Je$ is the range of
a completely positive conditional expectation $E$ with $\omega\circ E=\omega$, $\Je$ is invariant under the
modular group $\sigma^\omega$ and we have $\sigma^{\omega|_\Je}=\sigma^\omega|_\Je$,
\cite{takesaki1972conditional}.
By the results of \cite{junge2003noncommutative}, 
the crossed product $\mathcal T:=\Je\rtimes_{\sigma^{\omega|_\Je}} \mathbb R$
 becomes a subalgebra in $\mathcal R:=\Me\rtimes_{\sigma^\omega}\mathbb R$ and $E\otimes
 id_{L^2(\mathbb R)}$ is a conditional expectation on $\mathcal R$ with range $\mathcal
 T$. The canonical trace $\tau$  and the dual action $\theta$ on $\mathcal T$ are obtained by restriction from $\mathcal R$.
 It follows that the space $\tilde{\mathcal T}$ of $\tau$-measurable operators affiliated
 with
 $\mathcal T$ is a subspace in  $\tilde{\mathcal R}$, and $L^p(\Je)\subseteq
 L^p(\Me)$. Further, for  $p\ge 1$,  $L^p(\Je)$  is the range of a projection $E_p$ in
 $L^p(\Me)$ extending $E$ and the projections  satisfy the extended bimodule property
 \begin{equation}\label{eq:cp_bimodule}
E_s(hkl)=hE_q(k)l,\qquad h\in L^p(\Je),\ k\in L^q(\Me),\ l\in L^r(\Je),
 \end{equation}
with $p,q,r,s\in [1,\infty]$ and $\frac1p+\frac1q+\frac1r=\frac1s$. 

Since the subspace $h_\omega^{\frac1{2p}} \Je
h_\omega^{\frac1{2p}}$ is dense in $L^p(\Je)$ for all $p>0$ (see e.g. \cite[Lemma
A.1]{hiai2024alphaz}), we see that $L^p(\Je)$ coincides with our definition of
$L^p(\Je,\omega)$. The bimodule property \eqref{eq:cp_bimodule} shows that the
projections $E_p$ satisfy  \eqref{eq:Tp} and also proves  the statement \ref{it:lpJ5} in this case.
 Items \ref{it:lpJ2} and \ref{it:lpJ3} follow easily from the usual properties of Haagerup
$L_p$-spaces.  

Let $h$  and
$k$ be as in \ref{it:lpJ4}. Then $h\in L^{2q}(\Me)$ by \ref{it:lpJ1} and \ref{it:lpJ2}.
Using  the previous paragraphs, the statement of
\ref{it:lpJ4} can be equivalently formulated as:  $hkh$ is affiliated with $\mathcal T$
if and only if $k$ is. So assume that $hkh$ is affiliated with $\mathcal T$, and let $u$ be
any unitary in the commutant $\mathcal T'$. We then have
\[
hkh=u^*(hkh)u=(u^*hu)(u^*ku)(u^*hu)=h(u^*ku)h,
\]
where the last equality follows from $h\in L^{2q}(\Je)$ and therefore is affiliated with
$\mathcal T$. Since $h=h_\mu$ for a faithful normal state $\mu$, this implies that
$k=u^*ku$, for any $u\in \mathcal T'$, that is, {$k$ is affiliated with $\mathcal T$}. 
The converse statement in \ref{it:lpJ4} follows from \ref{it:lpJ3}.
This finishes the proof for a von Neumann subalgebra $\Je$.

\textbf{Case \ref{it:case-iia}.}~
We may assume that $\Me=\Be\oplus\Be^{op}$ and $\Je=\{a\oplus a^{op}\colon
a\in \Be\}$. Note that the involutive *-antiisomorphism  $op: \Be\to \Be^{op}$ extends to a
positive  isometry $op: L^p(\Be)\to L^p(\Be^{op})$ such that
$(h^{1/p}_\varphi)^{op}=h^{1/p}_{\varphi^{op}}$, $(x^*)^{op}=(x^{op})^*$ and $(xy)^{op}=y^{op}x^{op}\in L^r(\Be^{op})$ for $x\in
L^p(\Be)$ and $y\in L^q(\Be)$, $\frac1p+\frac1q=\frac1r$. It follows that we have 
\begin{equation}\label{eq:Lpcase2}
L^p(\Me)=L^p(\Be)\oplus L^p(\Be^{op})=\{x\oplus y^{op}\colon x,y\in L^p(\Be)\}.
\end{equation}
By Sect.~\ref{sec:structure}, the conditional expectation $E$ is given by \eqref{eq:CEs-iia} for an element $0\le \lambda\le1$ in $Z(\Be)$.
Let $\varphi\in \Me_*$, then $\varphi=\varphi_1\oplus\varphi_2^{op}$, with
$\varphi_1,\varphi_2\in \Be_*$. 
We have for $a\in \Be$
\[
\varphi\circ E(a\oplus 0)=\varphi(\lambda a\oplus \lambda a^{op})=(\varphi_1+\varphi_2)(\lambda a)
\]
and 
\[
\varphi\circ E(0\oplus a^{op})=\varphi((1-\lambda) a\oplus (1- \lambda)
a^{op})=(\varphi_1+\varphi_2)((1-\lambda) a).
\]
It follows that $\varphi$ is invariant under $E$ if and only if $\varphi_1=\tilde
\varphi(\lambda\,\cdot\,)$ and $\varphi_2=\tilde \varphi((1-\lambda)\,\cdot\,)$ for some
$\tilde\varphi\in \Be_*$, equivalently, 
$\varphi_1((1-\lambda)\,\cdot\,)=\varphi_2(\lambda\,\cdot\,)$. Using \ref{it:lpJ1}, we obtain
\begin{equation}\label{eq:L1case2}
L^1(\Je,\omega)=\{\lambda h\oplus (1-\lambda)h^{op}\colon h\in L^1(\Be)\}.
\end{equation}
Since $\omega\circ E=\omega$, we have $h_\omega=\lambda h_{\tilde\omega}\oplus
(1-\lambda)h_{\tilde\omega^{op}}$, for some faithful $\tilde\omega\in \Be_*^+$.
For any $a\oplus a^{op}\in \Je$ and $p>0$, we get
\[
h_\omega^{\frac1{2p}}(a\oplus
a^{op})h_\omega^{\frac1{2p}}=\lambda^{\frac1p}h_{\tilde\omega}^{\frac1{2p}}ah_{\tilde\omega}^{\frac1{2p}}\oplus
(1-\lambda)^{\frac1p}(h_{\tilde\omega}^{\frac1{2p}}ah_{\tilde\omega}^{\frac1{2p}})^{op},
\]
which implies that
\begin{equation}\label{eq:LpJcase2}
L^p(\Je,\omega)=\{\lambda^{\frac1p}h\oplus (1-\lambda)^{\frac1p}h^{op}\colon h\in
L^p(\Be)\}.
\end{equation}
From \eqref{eq:L1case2} and \eqref{eq:LpJcase2}, using polar decomposition in $L^p(\Be)$,
we obtain \ref{it:lpJ2}. The statement \ref{it:lpJ3} follows by direct computation from
\eqref{eq:LpJcase2}.  

Let now $h$ and $k$ be as in \ref{it:lpJ4}, then 
\[
h=\lambda^{\frac1{2q}}a\oplus
(1-\lambda)^{\frac1{2q}}a^{op},
\]
with $a=h_{\tilde\mu}^{\frac1{2q}}$, for some faithful $\tilde\mu\in \Be_*^+$, and
$k=x\oplus y^{op}$, $x,y\in L^p(\Be)$.
 Assume that 
\[
hkh=\lambda^{\frac1q}axa\oplus (1-\lambda)^{\frac1q}(aya)^{op}\in L^r(\Je,\omega).
\]
By \eqref{eq:LpJcase2}, there must be some $b\in L^r(\Be)$ such that 
\[
\lambda^{\frac1q}axa=\lambda^{\frac1q+\frac1p}b,\qquad
(1-\lambda)^{\frac1q}aya=(1-\lambda)^{\frac1q+\frac1p}b.
\]
Since $E$ is faithful, we have $\supp(\lambda)=\supp(1-\lambda)=1$.  This implies 
$axa=\lambda^{\frac1p}b$, $aya=(1-\lambda)^{\frac1p}b$, or 
\[
(1-\lambda)^{\frac1p}axa=\lambda^{\frac1p}aya.
\]
Since $\tilde \mu$ is faithful, this implies that
$(1-\lambda)^{\frac1p}x=\lambda^{\frac1p}y$. Taking the real and imaginary parts of this
operator, it is enough to assume that $x$ and $y$ are self-adjoint. Let $x=s|x|$ and $y=t|y|$ be the polar
decompositions in $L^{p}(\Be)$, then it follows 
that $s=t$ and $(1-\lambda)^{\frac1p}|x|=\lambda^{\frac1p}|y|$.
We thus have the polar decomposition $k=(s\oplus s^{op})(|x|\oplus |y|^{op})$, where
$|x|=h_{\rho}^{\frac1p}$, $|y|=h_\sigma^{\frac1p}$ for some $\rho,\sigma\in \Be_*^+$
satisfying $\rho((1-\lambda)\,\cdot\, )=\sigma(\lambda\,\cdot\,)$. This  implies that $\rho\oplus \sigma^{op}$ is invariant under $E$, in
other words, $|k|^p=h_{\rho\oplus \sigma^{op}}\in L^1(\Je,\omega)$. Since $s\oplus s^{op}
\in \Je$, we obtain $k\in L^p(\Je,\omega)$ by \ref{it:lpJ2}. The converse follows from
\ref{it:lpJ3}. 

To show \ref{it:lpJ5}, note that for any $a,b\in \Be$ and $p\ge 1$, we have
\[
h_\omega^{\frac1{2p}}(a\oplus b^{op})h_\omega^{\frac1{2p}}=\lambda^{\frac1p}h_{\tilde
\omega}^{\frac1{2p}}ah_{\tilde \omega}^{\frac1{2p}}\oplus (1-\lambda)^{\frac1p}h_{\tilde
\omega}^{\frac1{2p}}b^{op}h_{\tilde \omega}^{\frac1{2p}}
\]
and
\[
E_p(h_\omega^{\frac1{2p}}(a\oplus b^{op})h_\omega^{\frac1{2p}})=\lambda^{\frac1p}h_{\tilde
\omega}^{\frac1{2p}}(\lambda a+(1-\lambda)b)h_{\tilde \omega}^{\frac1{2p}}\oplus (1-\lambda)^{\frac1p}h_{\tilde
\omega}^{\frac1{2p}}(\lambda a+(1-\lambda)b)^{op}h_{\tilde \omega}^{\frac1{2p}}.
\]
It is easily seen that this map extends to
\begin{equation}\label{eq:ep_case2}
E_p(x\oplus y^{op})=(\lambda x+ \lambda^{\frac1p}(1-\lambda)^{\frac1{p'}}y)\oplus
(\lambda^{\frac1{p'}}(1-\lambda)^{\frac1p}x+(1-\lambda)y)^{op},\quad x,y\in L^p(\Be),
\end{equation}
here $\frac1p+\frac1{p'}=1$. The bimodule property \eqref{eq:p_bimodule} follows by a
direct computation from \eqref{eq:ep_case2}.

\textbf{Case \ref{it:case-iib}.}~
Let $\alpha:\Me\to \Me$ be an involutive *-antiisomorphism such that
$E=\frac12(id+\alpha)$ and $\Je$ is the set of fixed points of $\alpha$. 
Then $\alpha$ defines an involutive isometric positive map $\alpha_p$ on $L^p(\Me)$ for
each $p>0$ such that 
\begin{enumerate}[(1)]
\item $\alpha_\infty=\alpha$, $\alpha_1(h_\varphi)=h_{\varphi\circ\alpha}$,
\item \label{it:posit} $\alpha_p(x^*)=\alpha_p(x)^*$ and  for positive elements, $\alpha_p(h_\varphi^{1/p})=h_{\varphi\circ\alpha}^{1/p}$,
\item $\alpha_r(xy)=\alpha_q(y)\alpha_p(x)$, for $p,q,r$ as in \ref{it:lpJ4}  and $x\in L^p(\Me)$,
$y\in L^q(\Me)$.
\end{enumerate}
By the assumption, $\omega\circ\alpha=\omega$ and the above map coincides with our
definition of a map extension in \eqref{eq:Tp}. Indeed, using (2) and (3), we have
\[
\alpha_p(h_\omega^{\frac1{2p}}ah_\omega^{\frac1{2p}})=h_{\omega\circ\alpha}^{\frac1{2p}}\alpha(a)h_{\omega\circ\alpha}^{\frac1{2p}}=h_\omega^{\frac1{2p}}\alpha(a)h_\omega^{\frac1{2p}}.
\]
It follows that $E_p=\frac12(id+\alpha_p)$ extends $E$ for every $p>0$, and
$L^p(\Je,\omega)$ is the range of $E_p$, that is, precisely the set of
elements invariant under $\alpha_p$.

To prove \ref{it:lpJ2}, let  $h\in L_h^p(\Me)$, then there is some $\varphi\in
\Me_*^+$ and a symmetry $s\in \Me$ such that $h=sh_\varphi^{1/p}=h_\varphi^{1/p}s$.
Then $h\in L^p(\Je,\omega)$ if and only if $h=\alpha_p(h)=\alpha(s)\alpha_p(|h|)$. By the uniqueness of the polar
decomposition, this is equivalent to $s=\alpha(s)\in \Je$ and $|h|=\alpha_p(|h|)\in
L^p(\Je,\omega)$. By the property \ref{it:posit}, $h_\varphi^{1/p}\in L^p(\Je,\omega)$ if and only if
$\varphi\circ\alpha=\varphi$, which means that $|h|^p=h_{\varphi}\in L^1(\Je,\omega)$, by
\ref{it:lpJ1}. The rest of the statements follow easily by the properties of $\alpha_p$.

\textbf{Case \ref{it:case-iii}.}~
Let $\Je$ be of type $I_2$ and let $\tau$ be a faithful normal tracial
state on $\Me$ such that $E$ is the $\tau$-preserving conditional expectation. 

Since we have a faithful tracial state, we consider the conventional $L^p$-spaces $L^p(\Me,\tau)$ instead of the Haagerup $L^p$-spaces $L^p(\Me)$.
We refer to \cite{terp1981lpspaces} (see also \cite[Sec.~A.6]{hiai2018quantum}) for a discussion of the correspondence between the two, which we denote $L^p(\Me)\ni h \leftrightarrow \hat h \in L^p(\Me,\tau)$ here.
Importantly, the correspondence respects the multiplication $L^p \times L^q \to L^r$ with
$\frac 1p+\frac1q=\frac1r$ and the power structure. In particular, $(h^p)^\wedge =
\hat h^p$, $h\in L^1(\Me)^+$. Moreover, for $\varphi\in \Me_*$, $\hat h_\varphi$ is
determined as $\varphi(a)= \tau(\hat h_\varphi a)$, $a\in \Me$. 

The statement \ref{it:lpJ5} is now easy to prove. 
In the conventional $L^p$-spaces, we have $\Me\subset L^p(\Me,\tau)$ and for $p\ge 1$ the maps $\hat
E_p$, corresponding to $E_p$, coincide with $E$ on $\Me$. 
It follows that $\Je$ is dense in
$L^p(\Je,\omega)$ in this case.
For $\hat k \in \Me$, $\hat h,\hat l\in \Je$, the Jordan bimodule property
\eqref{eq:J-bimodule} of $E$ implies that the statement in \eqref{eq:p_bimodule} holds for the
corresponding elements $h$, $k$ and $l$.
Since the inclusions $\Me\subset L^p(\Me,\tau)$, and $\Je\subset L^p(\Je,\omega)$ are
norm-dense, continuity of the multiplication and the
maps $E_p$ extend the claims to arbitrary elements $k\in L^p(\Me)$,  $h\in
L^p(\Je,\omega)$ and $l\in L^r(\Je,\omega)$.

We will now show
that $L^p(\Je,\omega)$ is the closure of $\Je\subseteq \Me$ in $L^p(\Me)$, for all $p>0$. 
Indeed, since this is true for $L^1(\Je,\omega)$ by the previous paragraph, there is
a sequence $h_n\in \Je$ such that $\|h_n-\hat h_\omega\|_1\to 0$. By continuity of the absolute
value in $L^1(\Me)$, we may assume that $h_n\ge 0$ for all $n$. By \cite[Lem.~6]{kato2024onrenyi}, $\|h_n^{\frac1p}-\hat h_\omega^{\frac1p}\|_p\to 0$ for all $p>0$.
Using H\"older inequality, it follows that $h_n^{\frac1{2p}}ah_n^{\frac1{2p}}\to \hat
h_\omega ^{\frac1{2p}}a\hat h_\omega^{\frac1{2p}}$ in $L^p(\Me)$ for any $p>0$ and $a\in
\Je$, this implies that $\Je$ is dense in $L^p(\Je,\omega)$. By continuity of the
multiplication  $L^p \times L^q \times L^r\to L^s$, this implies \ref{it:lpJ3}. 

To show \ref{it:lpJ2} and \ref{it:lpJ4}, we further observe that for all $p>0$, 
$L^p(\Je,\omega)=\bar\Je\cap L^p(\Me)$, where $\bar \Je$ is the closure of $\Je$ in the
topology of convergence in measure induced by $\tau$ on $\Me$. 
Indeed, since the topology in  $L^p(\Me)$ is stronger than the measure topology, we have the inclusion $L_p(\Je,\omega)\subseteq \bar \Je\cap L^p(\Me)$. Conversely, it follows by the results of \cite[Thm.~3.5]{ayupov1984integration} that the self-adjoint part $\bar\Je_h$ of
$\bar\Je$ coincides with the set of all self-adjoint operators affiliated with $\Je$, in the sense that all the spectral projections are contained in $\Je$. But any such operator $h$ is in $L^p(\Je,\omega)$, since the sequence of truncated operators $h_n:=h\chi_{[-n,n]}(h)$ is contained in $\Je$ and converges to $h$ in $L^p(\Me)$. It follows that $L^p_h(\Je,\omega)=\bar\Je_h\cap L^p(\Me)$, which implies the statement.

Let now $h\in L^p_h(\Je,\omega)=\bar\Je_h\cap L^p(\Me)$, and let $h=\int \lambda
de_\lambda$ be the spectral decomposition. Put
\[
s=\mathrm{sgn}(h),\qquad |h|=\int |\lambda|de_\lambda,
\]
then $h=s|h|$ is the polar decomposition. Since $h$ is affiliated with $\Je$, we have
$e_\lambda\in \Je$ for all $\lambda$, so that $s\in \Je$ and $|h|\in \bar\Je_h$. It follows
that $|h|^p\in \bar\Je_h$ as well, and since $\tau(|h|^p)=\|h\|_p^p<\infty$, we have
$|h|^p\in \bar\Je\cap L^1(\Me)=L^1(\Je,\omega)$. For the converse, let $h=s|h|$ with $s\in
\Je$ and $|h|^p\in L^1(\Je,\omega)$, then similarly as before, there is some sequence $h_n\in \Je^+$ such 
 that $h_n\to |h|^p$ in $L^1(\Me)$. Using again \cite[Lem.~6]{kato2024onrenyi}, and
 continuity of the multiplication in $L^p(\Me)$, we get that 
 $s\circ h_n^{\frac1p}\to s\circ |h|=h$ in $L^p(\Me)$. Since $s\circ h_n^{1/p}\in \Je$ for
 all $n$, this proves \ref{it:lpJ2}. 

Finally, let $h$ and $k$ be as in \ref{it:lpJ4} and assume that $hkh\in L^r(\Je,\omega)$.
Note that if $r\ge 1$, then $p,q\ge 1$ as well and the statement follows from
\eqref{eq:p_bimodule} and the fact that $\mu$ is faithful. Let now $r<1$, and assume
first that $p\le 2$. Taking the real
and imaginary parts of $k$, it is clearly enough to assume that $k=k^*$, so that $k$ has
the polar decomposition 
$k=sl^{1/p}$ for some $l=h_\varphi\in L^1(\Me)^+$ and some symmetry $s$ commuting
with $l$. By the assumption and the above paragraph,   $hkh\in \bar \Je$. Since $h$ is affiliated with $\Je$, the truncations 
$h_n:=h \chi_{[\frac1n,n]}(h)$ are positive invertible operators in $\Je$, and we have
$h_n^{-1}h=e_n$, the corresponding spectral projections. Since $\mu$ is faithful,  
$e_n\nearrow 1$. We then have
\[
e_nke_n=h_n^{-1}(hkh)h_n^{-1}\in \bar \Je.
\]
Since we also have $e_nke_n\in L^p(\Me)$, it follows that $e_nke_n\in \bar\Je\cap
L^p(\Me)=L^p(\Je,\omega)$. It is now enough to show that $e_nke_n\to k$ in $L^p(\Me)$.
 We have
\[
\|k-e_nke_n\|_p=\|(1-e_n)k+e_nk(1-e_n)\|_p\le C(\|(1-e_n)k\|_p+\|e_nk(1-e_n)\|_p)
\]
(note that $\|\cdot\|_p$ may be only a  quasi-norm). We then write, with
$\frac1z=\frac1p-\frac12$,
\[
\|(1-e_n)k\|_p=\|(1-e_{n})l^{\frac12}l^{\frac1z}s\|_p\le
\|(1-e_{n})l^{\frac12}\|_2\|l^{\frac1z}s\|_z
=\varphi(1-e_n)^{\frac12} \|l^{\frac1z}s\|_z\to 0
\]
and similarly 
\[
\|e_nk(1-e_n)\|_p\le \|sl^{\frac1z}\|_z\varphi(1-e_n)^{\frac12} \to 0.
\]

If $p>2$, then, with  the dual parameter $p'$ with $\frac1p+\frac1{p'}=1$, there is some
$\tilde q>0$ such that $\frac1{p'}+\frac1{\tilde q}=\frac1q$. Putting
$k'=h_\mu^{\frac1{2p'}}kh_\mu^{\frac1{2p'}}$, we have
\[
hkh= h_\mu^{\frac1{2\tilde q}}k'h_\mu^{\frac1{2\tilde q}}, 
\]
and applying the previous arguments (with $p=1$), we obtain that $k'\in L^1(\Je,\omega)$. Using
\eqref{eq:p_bimodule}, this implies that $k\in L^p(\Je,\omega)$. The converse statement in
\ref{it:lpJ4} follows from \ref{it:lpJ3}.

\end{proof}

\section{$\alpha$-$z$-R\'enyi divergences and sufficient NUP maps}

The $\alpha$-$z$-R\'enyi divergences were introduced in finite-dimensional quantum information theory in \cite{audenaert_-z-renyi_2015} and were generalized to the context of von Neumann algebras
in \cite{kato2024onrenyi}. For any $\varphi,\psi\in \Me^+_*$, $\psi\ne 0$, and 
$\alpha,z>0$, $\alpha\ne 1$, we define  
\[
D_{\alpha,z}(\psi\|\varphi):=\frac1{\alpha-1}\log
\frac{Q_{\alpha,z}(\psi\|\varphi)}{\psi(1)},
\]
where
\[
Q_{\alpha,z}(\psi\|\varphi):=\begin{dcases} \Tr
\left(h_\varphi^{\frac{1-\alpha}{2z}}h_\psi^{\frac{\alpha}{z}}h_\varphi^{\frac{1-\alpha}{2z}}\right)^z, &
\text{if } 0<\alpha<1,\\[0.3em]
\|x\|_z^z, & \text{if } \alpha>1 \text{ and }
h_\psi^{\frac{\alpha}{z}}=h_\varphi^{\frac{\alpha-1}{2z}}xh_\varphi^{\frac{\alpha-1}{2z}}
\ \text{with}\\ & x\in s(\varphi)L^z(\Me)s(\varphi),\\[0.3em]
\infty,& \text{otherwise}.
\end{dcases}
\]
As important instances, this family includes the standard (or Petz) R\'enyi divergences 
$D_\alpha(\psi\|\varphi)=D_{\alpha,1}(\psi\|\varphi)$, and the sandwiched (or minimal)
R\'enyi divergences $\tilde D_{\alpha}(\psi\|\varphi)=D_{\alpha,\alpha}(\psi\|\varphi)$.

It was proved in \cite{hiai2024alphaz} that the $\alpha$-$z$-R\'enyi divergence satisfies the
data processing inequality: 
\begin{equation}\label{eq:dpi}
D_{\alpha,z}(\psi\circ T\|\varphi\circ T)\le D_{\alpha,z}(\psi\|\varphi)
\end{equation}
for any $\psi,\varphi\in \Me_*^+$ and any NUP map $T$, provided that the pair 
$(\alpha,z)$ is  contained in the region $R_{\text{DPI}}$ determined as
\begin{align}\label{eq:dpi_region<1}
 \max\{\alpha,1-\alpha\}\le z &\text{ if }  0<\alpha<1\\
\max\{\alpha/2,\alpha-1\}\le z\le \alpha &\text{ if } \alpha>1. \label{eq:dpi_region>1}
\end{align}
Moreover, it was shown that sufficiency of normal unital 2-positive maps can be characterized by
equality in \eqref{eq:dpi} for pairs of parameters in a region $R_{\text{rec}}$, strictly
contained in $R_{\text{DPI}}$.
Our aim in this section is to extend this statement to all
NUP maps, without the 2-positivity condition. As it was done in
\cite{luijk2026sufficiency} in the finite dimensional case, we will closely follow the
arguments of the proofs in \cite{hiai2024alphaz}, using our results on the spaces 
$L^p(\Je,\omega)$ obtained in  the previous section.
For this,  the following result will be our main tool. Note that the notion of a
sufficient NUP map naturally extends to positive normal functionals, so
that 
we do not need to assume the elements of $\Se$ to be normalized. 

\begin{lemma}\label{lemma:suff_alpha}
Let $\omega$ be a faithful normal state and let
$\mu\in \Me_*^+$. For $p,q>0$ and $\frac1r=\frac1p+\frac1q$, let $\rho\in \Me_*^+$ be such
that 
\[
h_{\rho}^{\frac1r}=h_{\omega}^{\frac1{2q}}h_\mu^{\frac1{p}}h_\omega^{\frac1{2q}}.
\]
 Let $T:\Ne\to \Me$ be an NUP map. Then  $T$
is sufficient with respect to $\{\mu,\omega\}$ if and only if $T$ is sufficient with respect to 
$\{\rho,\omega\}$.

\end{lemma}

\begin{proof} Let $\Ee=(\Me,\{\mu,\omega\})$ and let $E$ be the conditional expectation
onto the minimal sufficient JW*-subalgebra $\Je_{\Ee}$. Then $\omega\circ E=\omega$ and
$\mu\circ E=\mu$, so that by Lemma \ref{lemma:lpJ}, we obtain that $h_\rho^{\frac1r}\in
L^r(\Je_\Ee,\omega)$, which implies that also $\rho\circ E=\rho$. 
Conversely, let $\Fe=(\Me,\{\rho,\omega\})$ and let $F$ be the conditional expectation
onto $\Je_{\Fe}$. Then $\omega$ and $\rho$ are invariant under $F$ and Lemma \ref{lemma:lpJ}
implies that  $h_{\omega}^{\frac1{2q}}h_\mu^{\frac1{p}}h_\omega^{\frac1{2q}} \in
L^r(\Je_\Fe,\omega)$,  consequently, $h_\mu^{\frac1p}\in L^p(\Je_\Fe,\omega)$, so that
$\mu\circ F=\mu$. The statement now follows from Corollary \ref{coro:suff_ext}.

\end{proof}

\subsection{The sandwiched R\'enyi divergence}

Let us first summarize the results obtained in \cite{jencova2018renyi} for the  sandwiched R\'enyi divergence $\tilde D_\alpha$ with $\alpha>1$. Let
$\rho$ and $\omega$ be a pair of normal states on $\Me$, or more generally
$\rho,\sigma\in \Me_*^+$. Then $\tilde D_\alpha(\rho\|\omega)<\infty$ if and only if there
is some $\mu\in \Me_*^+$ such that
\[
h_\rho=h_\omega^{\frac{\alpha-1}{2\alpha}}h_\mu^{\frac1{\alpha}}h_\omega^{\frac{\alpha-1}{2\alpha}},
\]
in which case $\tilde
Q_\alpha(\rho\|\omega):=Q_{\alpha,\alpha}(\rho\|\omega)=\|h_\mu^{\frac1{\alpha}}\|_\alpha^\alpha=\mu(1)$.
Equivalently, $h_\rho$ is contained in the symmetric Kosaki $L^p$-space
$L^\alpha(\Me\|\omega)$
\cite{kosaki1984applications} and $\tilde
Q_\alpha(\rho,\omega)=\|h_\rho\|_{\alpha,\omega}^\alpha$, where
$\|\cdot\|_{\alpha,\omega}$ is the norm in $L^\alpha(\Me\|\omega)$. 
Let $T:\Ne\to \Me$ be an NUP map, then its preadjoint  defines a contraction
$T_*:L^\alpha(\Me\|\omega)\to L^\alpha(\Ne\|\omega\circ T)$, which directly leads to the data
processing inequality
\begin{equation}\label{eq:dpi_sandwiched}
\tilde Q_\alpha(\rho\circ T\|\omega\circ T)=\|T_*(h_\rho)\|^\alpha_{\alpha,\omega\circ T}\le
\|h_\rho\|^\alpha_{\alpha,\omega}=\tilde Q_\alpha(\rho\|\omega).
\end{equation}

The equality case in \eqref{eq:dpi_sandwiched} is easily characterized for 
$\alpha=2$. Namely, $L^2(\Me\|\omega)$ is a Hilbert space and the adjoint map of the
contraction $T_*$ is precisely the Petz recovery map  $(T_\omega)_*$. The following is
then an easy consequence of the properties of contractions on Hilbert spaces.

\begin{lemma}\cite[Cor.~3.17]{jencova2018renyi} \label{lemma:sandwiched_2} Let
$\rho,\omega \in \Me_*^+$ be such that $\tilde D_2(\rho\|\omega)<\infty$. Then 
$\tilde D_2(\rho\|\omega)=\tilde D_2(\rho\circ T\|\omega\circ T)$ if and only if
$\rho\circ T\circ T_\omega=\rho$, that is, $T$ is sufficient with respect to
$\{\rho,\omega\}$. 

\end{lemma}

For a general $\alpha>1$, equality in \eqref{eq:dpi_sandwiched} is related to the case
$\alpha=2$ by the following result, relying on interpolation techniques.

\begin{lemma}\label{lemma:eq_dpi_sandwiched} \cite[Lem.~4.5]{jencova2018renyi} Let $\mu,\omega\in \Me_*^+$ and let
$\alpha_0>1$. If the equality
\[
\|T_*(h_\omega^{\frac{\alpha-1}{2\alpha}}h_\mu^{\frac1{\alpha}}h_\omega^{\frac{\alpha-1}{2\alpha}})\|_{\alpha,\omega\circ
T}=\|h_\omega^{\frac{\alpha-1}{2\alpha}}h_\mu^{\frac1{\alpha}}h_\omega^{\frac{\alpha-1}{2\alpha}}\|_{\alpha,\omega}
\]
holds for $\alpha=\alpha_0$, then it holds for all $\alpha>1$.

\end{lemma}

We now prove the following extension of \cite[Thm.~4.6]{jencova2018renyi} to NUP maps that
are not necessarily 2-positive.

\begin{theorem}\label{thm:sufficiency_sandwiched} Let $\rho,\omega\in \Me_*^+$ and let
$\alpha>1$ be such
that $\tilde D_\alpha (\rho\|\omega)<\infty$. Then 
$\tilde D_\alpha(\rho\|\omega)=\tilde D_\alpha(\rho\circ T\|\omega\circ T)$ if and only if
 $T$ is sufficient with respect to
$\{\rho,\omega\}$. 

\end{theorem}

\begin{proof} 
Note that the assumption $\tilde D_{\alpha}(\rho\|\omega)<\infty$ implies that
$\mathrm{supp}(\rho)\le \mathrm{supp}(\omega)$. 
By restriction to the compressed algebras, we may assume that both $\omega$ and $\omega\circ T$ are  faithful. Moreover,
$h_\rho=h_\omega^{\frac{\alpha-1}{2\alpha}}h_\mu^{\frac1{\alpha}}h_\omega^{\frac{\alpha-1}{2\alpha}}$
for some $\mu\in \Me_*^+$. 

Assume  equality in the data processing inequality holds, that is,
$\|T^*h_\rho\|_{\alpha,\omega\circ T}=\|h_\rho\|_{\alpha,\omega}$.  By Lemma
\ref{lemma:eq_dpi_sandwiched}, this implies that also  
$\|T_*(h_\psi)\|_{2,\omega\circ
T}=\|h_\psi\|_{2,\omega}$, where 
$h_\psi=h_\omega^{\frac14}h_\mu^{\frac12}h_\omega^{\frac14}$. By Lemma
\ref{lemma:sandwiched_2}, this implies that $T$ is sufficient with respect to
$\{\psi,\omega\}$. The result follows by applying Lemma \ref{lemma:suff_alpha} (twice). 
The converse is clear from \eqref{eq:dpi_sandwiched}.

\end{proof}

\subsection{The $\alpha$-$z$-R\'enyi divergences}
\label{sec:recovery}

We are now ready to present the main results of the paper. As in \cite{hiai2024alphaz}, we
treat the cases $\alpha<1$ and $\alpha>1$ separately. In both cases, we relate the
equality in the data processing inequality \eqref{eq:dpi} to a similar equality for a
sandwiched R\'enyi divergence, and apply the previous paragraph to obtain sufficiency of the NUP map with respect to a modified pair of states. An application of Lemma
\ref{lemma:suff_alpha} then finishes the proof. 

\begin{theorem}\label{thm:alphale} Let $0 < \alpha < 1$ and $\max\{\alpha, 1 - \alpha\}
\le z$. Let $\rho,\omega \in  \Me_*^+$, $\rho\ne 0$, and assume either
that $\alpha < z$  and $\supp(\omega) \le \supp(\rho)$, or that $1 -\alpha < z$ and
$\supp(\rho) \le \supp(\omega)$. Then an NUP map $T : \Ne \to  \Me$  is sufficient  with
respect to $\{\rho, \omega\}$ if and only if
\[
D_{\alpha,z} (\rho \circ T\|\omega \circ T) = D_{\alpha,z} (\rho\|\omega). 
\]
\end{theorem}

\begin{proof} We will  assume that $\alpha < z$  and $\supp(\omega) \le \supp(\rho)$, the
other case will follow by exchanging the roles of $\rho$ and $\omega$, together with the
equality ${Q}_{\alpha,z}(\rho\|\omega)={Q}_{1-\alpha,z}(\omega\|\rho)$. 
As before, we may
assume that both $\rho$ and $\rho\circ T$ are faithful. Let us denote
$p:=\frac{z}{\alpha}$, $r:=\frac{z}{1-\alpha}$, so that we have $\frac1p+\frac1r=\frac1z$
and $p>1$ by the assumption. Let also $\mu,\sigma\in \Me_*^+$ be such that 
\[
h_\mu^{\frac1z}=h_\rho^{\frac1{2p}}h_\omega^{\frac1r}h_\rho^{\frac1{2p}},\qquad h_\sigma=
h_\rho^{\frac{p-1}{2p}}h_\mu^{\frac1p}h_\rho^{\frac{p-1}{2p}}.
\]
We then have
\[
Q_{\alpha,z}(\rho\|\omega)=\Tr\left(
h_\omega^{\frac1{2r}}h_\rho^{\frac1p}h_\omega^{\frac1{2r}}\right)^z=\Tr
\left(h_{\rho}^{\frac1{2p}}h_\omega^{\frac1r}h_\rho^{\frac1{2p}}\right)^z=\Tr
h_\mu=\|h_\mu^{\frac1p}\|_p^p=Q_{p,p}(\sigma\|\rho).
\]
It is shown in the proof of \cite[Thm.~4.5]{hiai2024alphaz} that if an NUP map $T$
satisfies the equality given in the theorem, then we have $\sigma\circ T\circ T_\rho=\sigma$. Note that
the proof of this fact is based on a variational formula for $Q_{\alpha,z}$ and uniform convexity
of  $L^p(\Me)$ with $p>1$, and it only uses contractivity of the maps $T_p$ and the Choi
inequality for $T$, in particular, 2-positivity of $T$ was not needed at this point.  
Hence we obtain that $T$ is sufficient with respect to the pair $\{\sigma,\rho\}$. Using
the definitions of $\sigma$ and $\mu$ and applying Lemma \ref{lemma:suff_alpha} (twice),
we obtain that $T$ is sufficient with respect to the original pair $\{\rho,\omega\}$.

\end{proof}

\begin{theorem}\label{thm:alphagr} Let $\alpha > 1$ and $\max\{\alpha/2, \alpha - 1\} \le
z\le \alpha$, and assume further that $\alpha < z + 1$. Let  $\rho,\omega \in \Me_*^+$
be such that $D_{\alpha,z} (\rho\|\omega) < \infty$. 
Then an NUP map $T : \Ne \to  \Me$  is sufficient  with
respect to $\{\rho, \omega\}$ if and only if
\[
D_{\alpha,z} (\rho \circ T\|\omega \circ T) = D_{\alpha,z} (\rho\|\omega). 
\]

\end{theorem}

\begin{proof} Let us put $p:=\frac{z}{\alpha}$ and $q:=\frac{z}{\alpha-1}$, so that this
time we have $p\in [1/2,1]$ and $q>1$.
From the assumption 
$D_{\alpha,z} (\rho\|\omega) < \infty$, it follows that $\supp(\rho) \le \supp(\omega)$,
so that we may again assume that $\omega$ and $\omega\circ T$ are both faithful.  Further,
there is some element $x\in L^z(\Me)$ such that
$h_\rho^\frac1p=h_\omega^{\frac1{2q}}xh_\omega^{\frac1{2q}}$, and it is easy to see that
such $x$ is necessarily positive (and unique), so that $x=h_\mu^{\frac1z}$ for some $\mu\in
\Me_*^+$. Let $\sigma\in \Me_*^+$ be such that
\[
h_\sigma=h_\omega^{\frac{q-1}{2q}}h_\mu^{\frac1q}h_\omega^{\frac{q-1}{2q}},
\]
then we have
\[
Q_{\alpha,z}(\rho\|\omega)=\|x\|_z^z=\mu(1)=Q_{q,q}(\sigma\|\omega).
\]
As shown in the proof of \cite[Thm.~4.7]{hiai2024alphaz}, equality in the data processing
inequality implies that $\sigma\circ T\circ T_\omega=\sigma$, and again this fact is
proved  without using the  assumption that $T$ is 2-positive. Hence $T$ is sufficient with
respect to $\{\sigma,\omega\}$. An application of Lemma \ref{lemma:suff_alpha} proves the
result.

\end{proof}

\begin{remark}\label{rem:region}  Let us note that the extra conditions on the parameter
values in Theorems \ref{thm:alphale} and \ref{thm:alphagr} determining the region
$R_{\text{rec}}$ within $R_{\text{DPI}}$ are essential and cannot be
dropped. Indeed, this follows from the examples known already in the finite dimensional
case:

In the case $\alpha<1$, \cite[Rem.~5.15]{hiai2017different} gives an example where 
$\supp(\rho)\le \supp(\sigma)$ and $T$ is a NUCP map, such that $T$ is not sufficient with
respect to $\{\rho,\sigma\}$, but the equality 
$D_{\alpha,1-\alpha}(\rho\|\sigma)=D_{\alpha,1-\alpha}(\rho\circ T\|\sigma\circ T)$ holds for any 
$\alpha\in (0,1)$. This shows that the condition $1-\alpha<z$ is necessary when
$\supp(\rho)\le \supp(\sigma)$. Necessity of the condition that $\alpha<z$ when $\supp(\sigma)\le
\supp(\rho)$ also follows, by exchanging $\rho\leftrightarrow\sigma$ and
$\alpha\leftrightarrow  1-\alpha$. If both states have the same support, the only
pair $(\alpha,z)$  forbidden by Theorem \ref{thm:alphale} is $\alpha=z=\frac12$. It was pointed out
already in
\cite[Appendix A.9]{mosonyi2015quantum} that Theorem \ref{thm:alphale} cannot be extended
to this case. 

In the case $\alpha>1$, the condition $\alpha< z+1$ is again necessary. In
\cite[Ex.~4.8]{hiai2017different}, an example of states with $\supp(\rho)\le
\supp(\sigma)$ and a NUCP map $T$ not sufficient with respect to $\{\rho,\sigma\}$ is given, such that 
$D_{\alpha,\alpha-1}(\rho\|\sigma)=D_{\alpha,\alpha-1}(\rho\circ T\|\sigma\circ T)$ holds for
$\alpha>1$. The limiting case $\lim_{\alpha\to \infty}\tilde D_\alpha$ was also excluded
in \cite[Appendix A.9]{mosonyi2015quantum}. 

\end{remark}

\printbibliography
\end{document}